\documentclass[conference]{IEEEtran}
\IEEEoverridecommandlockouts
\usepackage{cite}
\usepackage{dsfont}
\usepackage{amsmath,amssymb,amsfonts,amsthm}
\usepackage{algorithm}
\usepackage{algpseudocode}
\usepackage{graphicx}
\usepackage{textcomp}
\usepackage{xcolor}
\newtheorem{theorem}{Theorem}
\newtheorem{lemma}{Lemma}
\def\BibTeX{{\rm B\kern-.05em{\sc i\kern-.025em b}\kern-.08em
    T\kern-.1667em\lower.7ex\hbox{E}\kern-.125emX}}

\begin{document}

\title{Global Attitude Stabilization using Pseudo-Targets
}

\author{Mahathi Bhargavapuri$^1$\thanks{M. Bhargavapuri and S. R. Sahoo are with the Dept. of Electrical Engineering, Indian Institute of Technology Kanpur, Kanpur, India. \texttt{email: \{mahathi,srsahoo\}@iitk.ac.in}}
\and
Soumya Ranjan Sahoo$^1$
\and
Mangal Kothari$^2$\thanks{M. Kothari is with the Dept. of Aerospace Engineering, Indian Institute of Technology Kanpur, Kanpur, India. \texttt{email: mangal@iitk.ac.in}}}


\maketitle

\begin{abstract}
The topological obstructions on the attitude space of a rigid body make global asymptotic stabilization impossible using continuous state-feedback. This paper presents novel algorithms to overcome such topological limitations and achieve arbitrary attitude maneuvers with only continuous, memory-less state-feedback. We first present nonlinear control laws using both rotation matrices and quaternions that give rise to one almost globally asymptotically stabilizable equilibrium along with a nowhere dense set of unstable equilibria. The unstable equilibria are uniquely identified in the attitude error space. Pseudo-targets are then designed to make the controller believe that the attitude error is within the region of attraction of the stable equilibrium. Further, the pseudo-target ensures that maximum control action is provided to push the closed-loop system toward the stable equilibrium. The proposed algorithms are validated using both numerical simulations and experiments to show their simplicity and effectiveness.
\end{abstract}

\begin{IEEEkeywords}
Aerospace, nonlinear control, attitude stabilization, global stabilization, control applications.
\end{IEEEkeywords}

\section{Introduction}
\subsection{Motivation and Background}
The attitude tracking problem \cite{wen1991attitude} has been studied for many decades since it plays an important role in the control of a wide class of mechanical systems with rotational degrees of freedom such as satellites, aircrafts, underwater vehicles, hovercrafts, and robots \cite{bullo2000stabilization,sanyal2009inertia,seo2008high,lee2010geometric,tayebi2006attitude,sanyal2011AUV,leonard1997stability,wen1992motion}. For systems such as these, one would ideally want to achieve the desired attitude from any given arbitrary initial condition. Even though the Lie group structure of the attitude space (and consequently, the attitude error space) allows transforming a trajectory tracking problem into the easier problem of stabilization of the identity element \cite{maithripala2006almost}, achieving arbitrary aggressive attitude maneuvers remains a sizable challenge. This is because the topological properties of the rigid body attitude space prohibit the existence of an equilibrium that is globally asymptotically stable using continuous, memory-less state-feedback \cite{bhat}.

The rigid body attitude space is a boundary-less, compact manifold, represented globally and uniquely by the group of rotation matrices denoted by SO(3). In other words, SO(3) is not a vector space and does not have the property of contractibility. Continuous state-feedback controllers can at best achieve \textit{almost} global asymptotic attitude stabilization \cite{sanyal2009inertia,lee2010geometric,sanyal2011AUV,sanyal2008almost,lee2012exponential,maithripala2006almost}. The control action is lost when the attitude error lies on the unstable two-dimensional submanifold with the angular velocity being zero. Intuitively, this means that the controller is inactive when the attitude error between the current and desired attitudes is 180$^\circ$. However, such controllers can be augmented to achieve arbitrary, aggressive attitude maneuvers even when the attitude error lies on the unstable submanifold. This paper proposes the idea of providing pseudo-targets to ensure control action goes to zero only when the system reaches the stable equilibrium point.

While continuous state-feedback controllers fail to achieve global attitude stabilization, discrete control laws that do succeed are not robust to small measurement noise as noted in \cite{mayhew2011quaternion}. Hybrid, memory-based controllers are designed in  \cite{mayhew2009robust,mayhew2011quaternion,MAYHEW20131945} which use quaternions for attitude parameterization. Although these works successfully achieve global stabilization with robustness to small measurement noise and also overcome the \textit{unwinding phenomenon} \cite{bhat}, the method used cannot be generalized easily to any existing and simple control laws. Moreover, the implementation and stability proof seem tedious, and results directly on SO(3) cannot be found in the literature to the best of the authors' knowledge.

The present work overcomes the drawbacks of existing continuous, discrete, and hybrid controllers, and proposes generalized algorithms that generate pseudo-targets to enable simple continuous state-feedback to achieve arbitrary attitude maneuvers. Lyapunov-based analysis is utilized to ensure the stability of the proposed scheme. The proposed algorithms are generalized in the sense that they can be applied to any existing continuous state-feedback law designed directly on SO(3) or using unit quaternions to achieve global asymptotic stabilization. This is demonstrated in Section \ref{SecPseudo} of this paper. 

\subsection{Contributions}
This work attempts to provide a generalized scheme to drive a rigid body to the desired attitude from any given arbitrary initial condition. Nonlinear control laws on SO(3) as well as unit quaternions are designed to almost globally asymptotically stabilize the identity element of the attitude error space. The main contributions of this work are listed as follows:
\begin{itemize}
    \item The points in the error space other than the identity element where the continuous state-feedback becomes inactive are uniquely identified. 
    \item A pseudo-error generator is designed to ensure that whenever the attitude error is close to the set of points identified above, an intermediate target attitude is generated to ensure the controller behaves as though the error is within its region of attraction.
    \item It is ensured that the largest possible control input is provided when the error between the current and desired attitudes is 180$^\circ$.
\end{itemize}
The proposed algorithms can be applied to a wide class of continuous state-feedback controllers and become indispensable when a rigid body, like a satellite, is commanded to stabilize at arbitrary attitudes while using singularity-free controllers. Numerical simulations and experimental validation on a low-cost platform are described for cases with and without the proposed algorithm. The comparison allows one to observe the improvement with respect to what is currently achievable through continuous, memory-less state-feedback controllers. 

\section{Mathematical Preliminaries} \label{MathPrel}

This section contains the mathematical basis required to express the main ideas of the paper with clarity. To accentuate the importance of this work in a generalized setting, we use both the quaternion- and rotation matrix-based attitude representations. A unit quaternion is a function of four parameters subject to one constraint. Rotation matrices, on the other hand, have nine parameters subject to six constraints. For a lucid exposition on singularity-free attitude representation using unit quaternions and rotation matrices, the reader can refer to \cite{parwana2017quaternions,bullo2004geometric}.

Let $q \in \mathcal{S}^3$ be an unit quaternion that represents the attitude of a rigid body in $\mathds{R}^3$ with respect to an inertial frame, where $\mathcal{S}^3 = \{ q \in \mathds{R}^4 ~|~ q^T q = 1 \}$ denotes the three-dimensional unit sphere in $\mathds{R}^{4}$. The quaternion $q$ is composed of its scalar and vector parts given by $q = \left[q_{0} ~q_v^T \right]^T$, with $q_{0} \in \mathds{R}$ and $q_v \in \mathds{R}^3$, respectively. The quaternion multiplication $\otimes$ for $p$, $q$ $\in \mathcal{S}^3$ is defined as
\begin{equation}
p \otimes q = \begin{bmatrix}
p_0 q_0 - p_v^T q_v \\
p_0 q_v + q_0 p_v + \widehat{p}_v q_v
\end{bmatrix},
\end{equation}
where for any two vectors $x$, $y \in \mathds{R}^3$, $\widehat{x}y = x \times y$ denotes the vector cross product. The \textit{hat} map, $\widehat{\cdot} : \mathds{R}^3 \rightarrow \mathfrak{so}\textnormal{(3)}$, for a vector $x = \left[ x_1 ~x_2 ~x_3 \right]^T$ is given by
\begin{equation}
\widehat{x} = \begin{bmatrix}
0 & -x_3 & x_2 \\
x_3 & 0 & -x_1 \\
-x_2 & x_1 & 0
\end{bmatrix}.
\end{equation}
The quaternion kinematic equation is given by
\begin{equation}
\dot{q} =  \frac{1}{2} q \otimes \textbf{w}, 
\label{qderivative}
\end{equation}
where $\textbf{w} = \left[ 0 ~\omega^T \right]^T$ and $\omega \in \mathds{R}^3$ is the angular velocity in the body frame.
Expressing rotations in terms of unit quaternions is, however, non-unique since $q$ and $-q$ represent the same rotation in $\mathds{R}^3$. The configuration space of rigid body attitude is expressed uniquely in terms of rotation matrices which satisfies the properties of a Lie group. This group is called the special orthogonal group and denoted by SO(3) = $ \{ R \in \mathds{R}^{3\times3} ~|~ R^T R = RR^T= I, ~\mathrm{det}(R) = 1 \}$. Here, the rotation matrix $R$ transforms vectors in the body frame to the inertial frame. There exists a map $\mathcal{R} : \mathcal{S}^3 \rightarrow$ SO(3) such that $\mathcal{R}(q) = \mathcal{R}(-q)$ holds. The rotation kinematics are given by
\begin{equation}
\dot{R} = R \widehat{\omega}, 
\label{Rdot} 
\end{equation}
and the rigid body attitude dynamics are given by
\begin{equation}
\textbf{J}\dot{\omega} = - \widehat{\omega} \textbf{J}\omega + \textbf{M},
\label{rotdyn}
\end{equation}
where $\textbf{J} \in \mathds{R}^{3 \times 3}$ is the positive definite inertia tensor in the body frame and $\textbf{M} \in \mathds{R}^3$ is the vector consisting of externally applied moments.


Let $q_d$ denote the unit quaternion representing the desired rigid body attitude. The quaternion error is defined as
\begin{equation}
q_e = q_d^* \otimes q,
\end{equation}
where $q^* = \left[ q_0 ~-q_v^T \right]^T$ is the quaternion inverse. The configuration error on SO(3) for control design is chosen as a positive definite function $\Psi(R, R_d) : \textnormal{SO(3)} \times \textnormal{SO(3)} \rightarrow \mathds{R}$. An example for a candidate error function is 
\begin{equation}
\Psi(R, R_d) = \frac{1}{2} \mathrm{tr}\left\lbrace K \left( I - R_d^T R \right) \right\rbrace,\label{costrotmat}
\end{equation}
where $K = \textnormal{diag}(k_1, k_2, k_3)$ and $R_d$ is the desired attitude \cite{maithripala2006almost}. The eigenvalues of $K$ are chosen to be real, positive, and distinct to ensure $\Psi$ has only four critical points \cite{sanyal2009inertia,sahoo2014attitude}. 
The angular velocity error is defined as
\begin{equation}
e_\omega = \omega - R^T R_d \omega_d,
\label{omegaerr}
\end{equation}
where $\omega_d \in \mathds{R}^3$ is the desired angular velocity. The evolution of the errors $q_e$, $\Psi$, and $e_\omega$ with time is given by
\begin{subequations}
\begin{align}
\dot{q}_e &= \frac{1}{2}q_e \otimes \textbf{e}_\omega 
\label{quaterrdyn} \\
\dot{\Psi} &= e_R^T e_\omega\\
\textbf{J} \dot{e}_\omega &= -\widehat{\omega} \textbf{J} \omega + \textbf{M} + \textbf{J}\widehat{e}_\omega R^T R_d\omega_d - \textbf{J} R^T R_d \dot{\omega}_d,
\label{omegaerrdyn}
\end{align}
\end{subequations}
where $\textbf{e}_\omega = \left[ 0 ~e_\omega^T \right]^T$. The attitude error vector, $e_R \in \mathds{R}^3$, is defined using the variation of $\Psi$ with respect to $R$ which is given by
\begin{equation}
\begin{aligned}
\textbf{D}_R \Psi (R, R_d)\cdot \delta R &= -\frac{1}{2}\mathrm{tr}\left( K R_d^T R\widehat{\eta} \right) \\
    &= \frac{1}{2} \left( K R_d^T R - R^T R_d K\right)^\vee \cdot \eta \\
	&= e_R ^T \eta,
\end{aligned}
\label{psiVar}
\end{equation}
where $\delta R = \frac{d}{d \epsilon}\vert_{_{_{\epsilon=0}}} R ~\textnormal{exp} (\epsilon \widehat{\eta}) = R\widehat{\eta}$ is the infinitesimal variation in $R$, $\eta \in \mathds{R}^3$, and $\cdot^\vee : \mathfrak{so}\textnormal{(3)} \rightarrow \mathds{R}^3$ is the inverse of hat map, called the \textit{vee} map.   

\section{Global Attitude Stabilization using Pseudo-Targets} \label{SecPseudo}
The existence of multiple equilibria in the attitude error space makes global stabilization of its identity element using continuous controllers impossible. Essentially, the control action vanishes at multiple points including the identity element of the error space. Understanding this behaviour of the closed-loop system and characterization of the unstable equilibria forms the core motivation and basis to design pseudo-targets for global asymptotic stabilization. Hence, nonlinear control laws in terms of quaternions and rotation matrices that can achieve almost global stabilization are designed first. The following theorems discuss the Lyapunov-based control design.

\subsection{Almost Globally Stabilizing Control Laws}

Most quaternion-based controllers found in literature ignore the fundamental problem of double covering the attitude space \cite{wen1991attitude,fresk2013full,seo2008high,tayebi2008unit}. This can lead to the unwinding phenomenon and is highly undesirable. The following theorem provides a method to overcome this obstruction.

\begin{theorem}
Given a smooth command $q_d(t)$, for positive constants $k_q$, $k_{\omega q} \in \mathds{R}$, the control law given by
\begin{equation}
\begin{aligned}
\textnormal{\textbf{M}} = - k_q q_{e0} q_{ev} - k_{\omega q} e_\omega + \widehat{\omega} \textnormal{\textbf{J}} \omega - \textnormal{\textbf{J}}\widehat{e}_\omega R^T R_d\omega_d \\ + \textnormal{\textbf{J}} R^T R_d \dot{\omega}_d,
\end{aligned}
\label{Moment_q}
\end{equation} 
almost globally asymptotically stabilizes identity element of the error space, $(q_e^{id}, ~e_\omega^{id}) \equiv \left( \left[ \pm 1 ~0_{3 \times 1}^T \right]^T, ~0_{3 \times 1} \right)$, and the region of attraction is given by $\mathcal{S}^3 \backslash \mathcal{A} \times \mathds{R}^3$, where $\mathcal{A} = \left\lbrace q_e \in \mathcal{S}^3 ~|~ q_{e0} = 0 \right\rbrace$. 
\label{Prop1}
\end{theorem}
\begin{proof}
Consider the following Lyapunov function candidate
\begin{equation}
V_q = k_q \left( 1 - q_{e0}^2 \right) + \frac{1}{2} e_\omega ^T \textnormal{\textbf{J}} e_\omega.
\end{equation}
The function $V_q$ is positive definite about the identity element $(q_e^{id}, ~e_\omega^{id}) \equiv \left( \left[ \pm 1 ~0_{3 \times 1}^T \right]^T, ~0_{3 \times 1} \right)$. The time derivative of $V_q$ is given by
\begin{equation}
\dot{V}_q = k_q q_{e0} q_{ev}^T e_\omega + e_\omega^T \textbf{J} \dot{e}_\omega
\label{B2}
\end{equation}
Substituting \eqref{omegaerrdyn} in \eqref{B2}, the expression for $\dot{V}_q$ can be written as
\begin{equation}
\begin{aligned}
\dot{V}_q = e_\omega^T ( k_q q_{e0} q_{ev} - \widehat{\omega} \textbf{J} \omega + \textbf{M} + \textbf{J}\widehat{e}_\omega R_e^T\omega_d - \textbf{J} R_e^T \dot{\omega}_d ).
\end{aligned}
\label{vdot_quat}
\end{equation}
Further, substituting \eqref{Moment_q} in \eqref{vdot_quat} results in 
\begin{equation}
\begin{aligned}
\hspace{1cm} \dot{V}_q = - k_{\omega q} e_\omega^T e_\omega \leq 0, \hspace{1cm}
\end{aligned}
\end{equation}
and the closed-loop system becomes
\begin{equation}
\textbf{J}\dot{e}_\omega = - k_{\omega q} e_\omega - k_q q_{e0} q_{ev}.
\label{CL_q_dyn}
\end{equation}
Note that the set of equilibria for the closed-loop system in \eqref{CL_q_dyn} is given by $\left\lbrace (q_e^{id}, ~e_\omega^{id})\right\rbrace \cup (\mathcal{A} \times \left\lbrace e_\omega^{id}\right\rbrace ) $. The existence of multiple equilibria implies global stability is non-viable. The next best possible scenario is to ensure that the desired equilibrium, $(q_e^{id}, ~e_\omega^{id})$, is \textit{almost} globally asymptotically stable. We have $\mathcal{M} = \left\lbrace (q_e^{id}, ~e_\omega^{id}) \right\rbrace$ as the largest invariant set in $\mathcal{S}^3 \backslash \mathcal{A} \times \mathds{R}^3$. Asymptotic stabilization of the equilibrium $(q_e^{id}, ~e_\omega^{id})$ in $\mathcal{S}^3 \backslash \mathcal{A} \times \mathds{R}^3$ using the control law in \eqref{Moment_q} follows from \textit{LaSalle's Invariance Principle}. The set $\mathcal{A}$ forms a two-dimensional submanifold from which $(q_e^{id}, ~e_\omega^{id})$ cannot be stabilized, and hence the control law in \eqref{Moment_q} can achieve almost global asymptotic stability. 
\end{proof}

For the generation of pseudo-targets on SO(3), it is essential that the controller ensures minimum possible closed-loop equilibria. The following theorem provides one such control law which has been derived using the results for a general class of systems evolving on Lie groups \cite{maithripala2006almost}. 

\begin{lemma}
Given a smooth command $R_d(t)$, for positive constants $k_R$, $k_{\omega_R} \in \mathds{R}$, the control law given by
\begin{equation}
\begin{aligned}
\textnormal{\textbf{M}} = - k_R e_R - k_{\omega_R} e_\omega + \widehat{\omega} \textnormal{\textbf{J}} \omega - \textnormal{\textbf{J}} \widehat{e}_\omega R^T R_d\omega_d + \textnormal{\textbf{J}} R^T R_d \dot{\omega}_d,
\end{aligned}
\label{Moment_R}
\end{equation} 
almost globally asymptotically stabilizes identity element of the error space, $(e_R^{id}, ~e_\omega^{id}) \equiv \left( 0_{3 \times 1}, ~0_{3 \times 1} \right)$, and the region of attraction is given by $\textnormal{SO(3)}\backslash\mathcal{Y} \times \mathds{R}^3$, where $\mathcal{Y} = \left\lbrace \textnormal{diag}(1, -1, -1), ~\textnormal{diag}(-1, 1, -1), ~\textnormal{diag}(-1, -1, 1) \right\rbrace$. 
\label{Prop2}
\end{lemma}
\begin{proof}
Consider the Lyapunov function candidate \begin{equation}
V_R = k_R \Psi(R_e) + \frac{1}{2}e_\omega ^T \textnormal{\textbf{J}} e_\omega,
\end{equation}
whose time derivative is
\begin{equation}
\dot{V}_R = k_R e_R^T e_\omega + e_\omega^T \left( - \widehat{\omega} \textbf{J} \omega + \textbf{M} + \textbf{J}\widehat{e}_\omega R_e^T\omega_d - \textbf{J} R_e^T \dot{\omega}_d \right).
\end{equation}
From the control law in \eqref{Moment_R}, we obtain $\dot{V}_R = - k_{\omega_R} e_\omega^T e_\omega \leq 0$.
The closed-loop system thus becomes
\begin{equation}
\textbf{J}\dot{e}_\omega = - k_{\omega_R} e_\omega - k_R e_R.
\label{CL_R_dyn}
\end{equation}
The largest invariant set in $\textnormal{SO(3)} \backslash \mathcal{Y} \times\mathds{R}^3$ contains only the equilibrium, $(e_R^{id}, ~e_\omega^{id}) \equiv \left( 0_{3 \times 1}, ~0_{3 \times 1} \right)$. From \textit{LaSalle's Invariance Principle}, the equilibrium $(e_R^{id}, ~e_\omega^{id})$ is asymptotically stable within $\textnormal{SO(3)} \backslash \mathcal{Y} \times \mathds{R}^3$. The only remaining critical points of $\Psi(R_e)$ lie in $\mathcal{Y}$ and form a nowhere dense set. 
Hence, the control law in \eqref{Moment_R} almost globally asymptotically stabilizes the equilibrium $(e_R^{id}, ~e_\omega^{id})$. 
\end{proof}
\begin{figure}[ht]
\begin{center}
\includegraphics[width=0.44\textwidth]{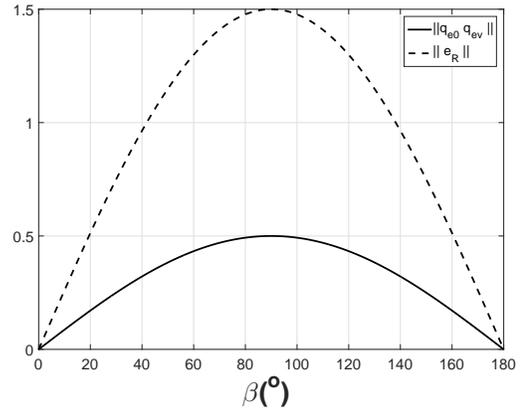} 
\caption{Variation of $\Vert q_{e0}q_{ev} \Vert$ and $\Vert e_R \Vert $ with error in angle $\beta$ varying from 0$^\circ$ to 180$^\circ$ about the body axis $e_3 = \left[ 0 ~0 ~1 \right]^T$, $K = \textnormal{diag}\left( 1, 2, 3 \right)$.}
\label{error_norm}                                    
\end{center}                                	  	  
\end{figure}
From \eqref{Moment_q} and \eqref{Moment_R}, it is clear that the control action depends directly on $\Vert q_{e0}q_{ev}\Vert$ and $\Vert e_R \Vert$, where $\Vert \cdot \Vert$ is the 2-norm defined on $\mathds{R}^3$. When the attitude error is close to 180$^\circ$ about some axis of rotation, and assuming $\omega = \omega_d = 0$ for stabilization, the control actions vanish to zero. This is shown in Fig. \ref{error_norm}. This means that the control effort is negligible at large attitude errors, which is undesirable. 


\subsection{Quaternion-based Pseudo-Target}
For any error quaternion in $\mathcal{A}$, the statement $\Vert q_{ev} \Vert = 1$ always holds due to the property of unit quaternions. Whenever $q_e$ lies in $\mathcal{A}$, an intermediate element is generated such that $q_e^{int} = \left[\pm 1 ~q_{ev}^T  \right]^T$. Note that $q_e^{int} \notin \mathcal{S}^3$. The pseudo-error quaternion, $q_e^{pseudo}$, is computed as
\begin{equation}
q_e^{pseudo} = \frac{q_e^{int}}{\Vert q_e^{int} \Vert}.
\label{qPseudo}
\end{equation}
This ensures that $ \Vert q_{e0}^{pseudo} q_{ev}^{pseudo} \Vert = 0.5$ always and provides maximum proportional action (see Fig. \ref{error_norm}). 

\begin{algorithm}
\caption{Pseudo-error quaternion generation}\label{alg:quat}
\begin{algorithmic}[1]
\Function{Error quaternion}{$q_d,q$}
\State $q_e\gets q_d^{*} \otimes q$ \Comment{The original error quaternion}
\While{$ q_{e} \in \mathcal{Q} $}	
\State $q_{e0}\gets \pm 1$ \Comment{Forcing the scalar part $q_{e0}$ to $\pm 1$}
\State Generate $q_e^{int} = \left[ q_{e0} ~q_{ev}^T \right]^T$ 
\State $q_e^{pseudo} \gets \eqref{qPseudo}$
\State $q_e \gets q_e^{pseudo}$
\EndWhile
\State \textbf{return} $q_e$\Comment{$q_e$ to be used by the controller}
\EndFunction
\end{algorithmic}
\end{algorithm}

The same logic can be extended to a small region around $q_{e0} = 0$ to improve proportional action of the controller. Define a set $\mathcal{Q}= \left\lbrace q_e \in \mathcal{S}^3 ~|~ \vert q_{e0} \vert < \varepsilon \right\rbrace $ where $\varepsilon \in \mathds{R}$ is small and positive. Modifying the algorithm to include any $q_e \in \mathcal{Q}$, a pseudo-code for generating the pseudo-error quaternion is shown in Algorithm \ref{alg:quat}. Note that the direction in which the pseudo-target is generated is irrelevant since the desired rotation is already close to 180$^\circ$.

\subsection{Pseudo-Targets on SO(3)}

Creating pseudo-targets on SO(3) is not as straightforward as using quaternions. Unlike quaternions, where the axis of rotation is $q_{ev}$ when error is 180$^\circ$ and the pseudo-error is weighted accordingly, the configuration error $\Psi$ plays an important role in identifying the axis about which the error is close to 180$^\circ$. Choosing $\Psi$ as in \eqref{costrotmat} ensures the minimum number of critical points of such an error function on SO(3). Let $e_1 = \left[ 1 ~0 ~0 \right]^T$, $e_2 = \left[ 0 ~1 ~0 \right]^T$, and $e_3 = \left[ 0 ~0 ~1 \right]^T$ be the body frame axes and let $R_e = R_d^T R$. The error matrix, $^{1}R_e = \textnormal{diag}(1, -1, -1)$, corresponds to an error of 180$^\circ$ about $e_1$. Similarly, $^2R_e = \textnormal{diag}(-1, 1, -1)$ and $ ^3R_e = \textnormal{diag}(-1, -1, 1)$ correspond to errors about $e_2$ and $e_3$, respectively. Corresponding to each of these unstable equilibria, $\Psi$ assumes a unique value that depends on the elements of $K$. Specifically, $\Psi = k_2 + k_3$ for $R_e = $ $^1R_e$, $\Psi = k_1 + k_3$ for $R_e =$ $^2R_e$, and $\Psi = k_1 + k_2$ for $R_e = $ $^3R_e$. Since $\Vert e_R \Vert$ reaches its maximum when the error is $\pm90^\circ$ about the axis of rotation (see Fig. \ref{error_norm}), the pseudo-error rotation matrix, $R_e^{pseudo}$, is chosen accordingly. Hence, for $R_e = $ $^1R_e$, we have
\begin{equation}
R_e^{pseudo} = \begin{bmatrix}
1 & 0 & 0 \\
0 & 0 & -1 \\
0 & 1 & 0
\end{bmatrix}.
\label{R-pseudo1}
\end{equation}
In a similar fashion, $R_e^{pseudo}$ for the remaining two critical points is chosen such that it corresponds to 90$^\circ$ error about the body axes $e_2$ and $e_3$.

\begin{algorithm}
\caption{Pseudo-error rotation matrix generation}\label{alg:SO3}
\begin{algorithmic}[1]
\Function{Error vector $e_R$}{$R_d,R$}
\State $\Psi(R_e)\gets$ \eqref{costrotmat} \Comment{Configuration error on SO(3)}
\State $e_R \gets \frac{1}{2} \left( KR_e - R_e^TK^T \right)^\vee $
\If{$ R_{e} \in \bar{\mathcal{S}}_1 $}	
\State $R_e^{pseudo} \gets \eqref{R-pseudo1}$
\State Compute intermediate $e_R$
\EndIf
\If{$ R_{e} \in \bar{\mathcal{S}}_2 $}	
\State $R_e^{pseudo} \gets R_e$ given by 90$^\circ$ error about $e_2$
\State Compute intermediate $e_R$
\EndIf
\If{$ R_{e} \in \bar{\mathcal{S}}_3 $}	
\State $R_e^{pseudo} \gets R_e$ given by 90$^\circ$ error about $e_3$
\State Compute intermediate $e_R$
\EndIf
\State \textbf{return} $e_R$\Comment{$e_R$ to be used by the controller}
\EndFunction
\end{algorithmic}
\end{algorithm}

Improvement in proportional action is achieved by defining three sets corresponding to the three critical points. For errors close to $^1R_e$, we define $\bar{\mathcal{S}}_1 = \left\lbrace R_e \in \textnormal{SO(3)} ~\vert ~| \Psi(R_e) - (k_2 + k_3) | < \varepsilon \right\rbrace$. The sets $\bar{\mathcal{S}}_2 = \left\lbrace R_e \in \textnormal{SO(3)} ~\vert ~| \Psi(R_e) - (k_1 + k_3) | < \varepsilon \right\rbrace$ and $\bar{\mathcal{S}}_3 = \left\lbrace R_e \in \textnormal{SO(3)} ~\vert ~| \Psi(R_e) - (k_1 + k_2) | < \varepsilon \right\rbrace$ are analogously defined for errors close to $^2R_e$ and $^3R_e$, respectively. The pseudo-code for generating $R_e^{pseudo}$ is given in Algorithm \ref{alg:SO3}.

\section{Results} \label{Res}
Numerical simulations in MATLAB as well as experimental validation of the proposed scheme are shown in this section. Since the stability properties of the nonlinear controller are untouched, the response is smooth. 

For the purpose of simulations, the inertia matrix is taken to be $\textbf{J} = \textnormal{diag}\left( 0.0125, 0.0125, 0.025 \right)$. The gains in \eqref{Moment_q} were chosen to be $k_q = 10$ and $k_{\omega q} = 1.5$. The gains in \eqref{Moment_R} were taken to be $k_R = 5$ and $k_{\omega_R} = 2.1$ with $K = \textnormal{diag}\left( 1, 2, 3 \right)$. A large attitude maneuver of 0$^\circ$ to 180$^\circ$ is attempted about $e_3$ with $\varepsilon = 0.01$. For the stabilization case, we have $\omega = \omega_d = 0$. In terms of quaternions, it is represented as a change from an initial state of $q = \left[ \pm 1 ~0 ~0 ~0\right]^T$ to the desired state $q_d = \left[ 0 ~0 ~0 ~\pm1\right]^T$. 
\begin{figure}[ht]
\begin{center}
\includegraphics[width=0.48\textwidth]{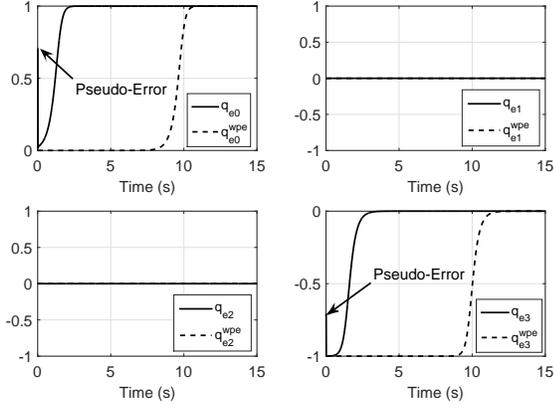}    
\caption{Quaternion error components when $ q_e \in \mathcal{A}$. $q_e$ denotes the error quaternion when pseudo-error is used. $q_e^{wpe}$ stands for quaternion error without implementing the pseudo-error generation algorithm.}									   
\label{quat_w_pseudo}                                 
\end{center}                                 
\end{figure}

Stabilization of the quaternion error to $q_e^{id}$ with and without pseudo-error is shown in Fig. \ref{quat_w_pseudo}. It is observed that in the presence of small measurement noise, stabilization may take arbitrarily long time, in this case, $11s$. This is expected as the vector field resulting from \eqref{Moment_q} vanishes when the error is close to the set $\mathcal{A}$. This time can be random and depends on the magnitude of noise. Using pseudo-errors it is ensured that closed-loop response of the system is similar to that when $q_e \in \mathcal{S}\backslash \mathcal{A}$.




\begin{figure}[ht]
\begin{center}
\includegraphics[width=0.48\textwidth]{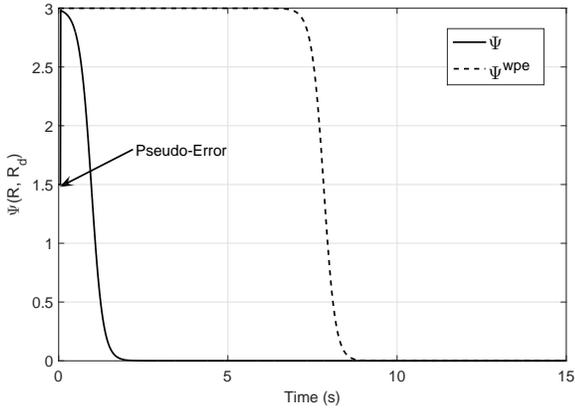}    
\caption{Variation of $\Psi(R, R_d)$ with time when the error is 180$^\circ$ about the body axis $e_3$. $\Psi$ denotes the configuration error with pseudo-error and $\Psi^{wpe}$ is the configuration error without implementing the pseudo-error generation algorithm.}									   
\label{Psi_w_pseudo}                                 
\end{center}                                 
\end{figure}

\begin{figure}[ht]
\begin{center}
\includegraphics[width=0.48\textwidth]{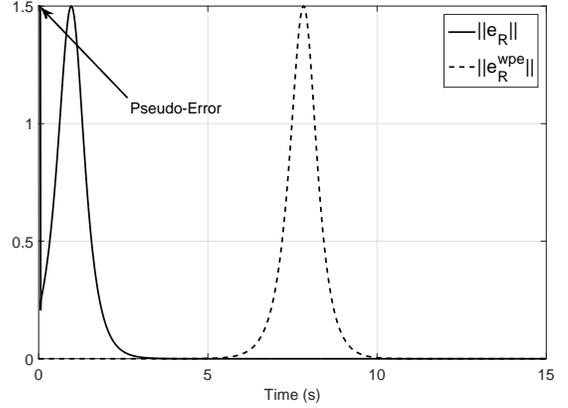}    
\caption{Variation of $ \Vert e_R \Vert $  and $ \Vert e_R^{wpe} \Vert $ with time when the error is 180$^\circ$ about the body axis $e_3$. $e_R$ denotes the rotation error when pseudo-error is used. $qe_R^{wpe}$ stands for rotation error without implementing the pseudo-error generation algorithm.}									   
\label{eR_w_pseudo}                                 
\end{center}                                 
\end{figure}

Numerical simulations for the above mentioned attitude maneuver is shown using the control law in \eqref{Moment_R}. The initial and desired attitude of the rigid body are $R = I$ and $R_d = \textnormal{diag}\left( -1,-1,1\right)$, respectively. The variation of $\Psi$ and $\Vert e_R \Vert$ with and without pseudo-error is shown in Fig. \ref{Psi_w_pseudo} and Fig. \ref{eR_w_pseudo}, respectively. Observe that the configuration error goes to zero after an arbitrarily long time when pseudo-error is not used.

Experiments were performed to validate the proposed method using a miniature unmanned aerial vehicle with parameters identical to those used for simulations. To ensure brevity, results for the case of quaternions are presented. A flip-switch is used to provide the desired attitude at $t=12s$ with pseudo-error and at $t=18s$ while pseudo-error is disabled. The results in Fig. \ref{quat_expt_w_pseudo} are analogous to those in Fig. \ref{quat_w_pseudo}. Video of the experiment can be found at \textbf{https://youtu.be/YwthUIjQqiI} .

\begin{figure}[ht]
\begin{center}
\includegraphics[width=0.48\textwidth]{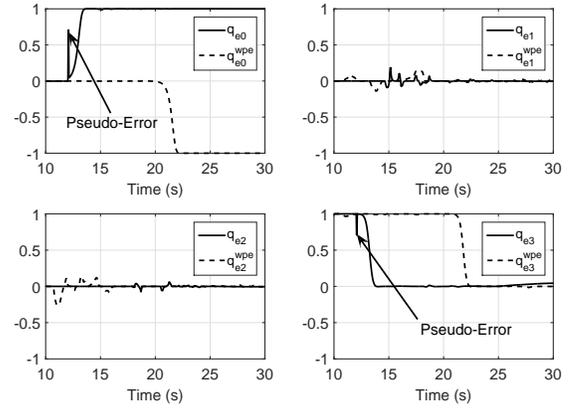}    
\caption{Experimental results for a 180$^\circ$ error about the body axis $e_3$ with and without pseudo-error. $q_e$ denotes the error quaternion when pseudo-error is used. $q_e^{wpe}$ stands for quaternion error without implementing the pseudo-error generation algorithm.}									   
\label{quat_expt_w_pseudo}                                 
\end{center}                                 
\end{figure}

\section{Conclusion} \label{Conclusions}
The idea of providing a pseudo-target whenever the attitude error is 180$^\circ$ about any arbitrary axis of rotation is explored. Such errors are uniquely identified in terms of error quaternions and configuration error functions on SO(3). Almost globally stabilizing continuous control laws were examined and their inaction at the critical points of the configuration space was shown. Algorithms to generate pseudo-errors were presented to overcome topological obstructions in the attitude space. Numerical simulations that show the effectiveness of the proposed scheme were discussed along with experimental validation. 

\section*{Acknowledgment}

The authors would like to thank Venkata Ramana Makkapati for his invaluable suggestions during the preparation of this manuscript.

\bibliographystyle{IEEEtran}
\bibliography{autosam}

\end{document}